\documentclass[submission,copyright,creativecommons]{eptcs}

\usepackage{iftex}

\usepackage{amsmath,amssymb,amsthm,mathrsfs}
\usepackage{stmaryrd}
\usepackage{iftex}
\usepackage{tabularx}
\usepackage[ruled, vlined]{algorithm2e}
\usepackage{tikz}
\usepackage{graphicx}

\usepackage[most]{tcolorbox}
\usetikzlibrary{positioning,arrows.meta}

\newtheorem{theorem}{Theorem}
\newtheorem{lemma}[theorem]{Lemma}

\newtheorem{corollary}[theorem]{Corollary}
\newtheorem{claim}[theorem]{Claim}
\theoremstyle{definition}

\renewcommand{\phi}{\varphi}

\ifpdf
  \usepackage{underscore}         
  \usepackage[T1]{fontenc}        
\else
  \usepackage{breakurl}           
\fi

\title{The Model Checking Problem for Distributed Knowing How is $\Delta_2^p$-Complete}

\author{
Ziqi Wang
\institute{ILLC, University of Amsterdam}
\email{ziqi.wang6@student.uva.nl}
\and
Ronald de Haan
\institute{ILLC, University of Amsterdam}
\email{me@ronalddehaan.eu}
}
\def\titlerunning{The Model Checking Problem for Distributed Knowing How is $\Delta_2^p$-Complete}
\def\authorrunning{Wang \& de Haan}
\begin{document}
\maketitle

\begin{abstract}
We investigate the complexity of the model checking problem for distributed knowing how. We show that the problem is $\Delta_2^p$-complete.
\end{abstract}

\section{Introduction}

Knowing-how logics study agents' abilities to achieve goals by suitable actions. This makes them particularly relevant to artificial intelligence, especially to planning, where the main question is whether a desired objective can be guaranteed by some plan. The logic of distributed knowing how \cite{DKH2025} generalizes two existing frameworks of knowing how. The planning-based framework, initiated in \cite{Wang2015, Wang2018}, treats knowing how as the existence of a multi-step plan that can be executed, always terminates, and guarantees the goal, while the coalition-based framework, initiated in \cite{TogetherWeKnowHow2017, TogetherWeKnowHow2018}, treats group know-how as the existence of a one-step joint action that ensures the goal.

In this work, we prove the model checking problem for the logic of distributed knowing how is ${\Delta_2^p}$-complete. More precisely, we introduce a model checking algorithm in $\Delta_2^p$, the class of problems decidable in polynomial time with polynomially many calls to an $\mathsf{NP}$ oracle. The main difficulty lies in the $Kh_G$-modality. Instead of explicitly constructing the distributed action set $A_G^*$, our procedure first computes the satisfaction set of the inner formula and then runs a fixpoint algorithm over the quotient set $S/{\sim_G}$. We show that the problem is ${\Delta_2^p}$-hard, by a reduction from \textsc{SNSAT} \cite{SNSAT}. Together both results, we obtain a tight bound for this problem.

\section{Preliminaries}

Let $\mathcal{P}$ be a denumerable set of proposition symbols and let $I$ be a finite set of agents. The language $\mathsf{DKH}$ is defined by the grammar
\[
\varphi ::= \top \mid p \mid \neg\varphi \mid (\varphi \wedge \varphi) \mid K_G\varphi \mid Kh_G\varphi
\]
where $p\in \mathcal{P}$ and $G\subseteq I$. A model $\mathcal{M}$ is a tuple $\left\langle
S,\{\sim_i\}_{i\in I},\{A_G\}_{G\subseteq I},
\left\{\xrightarrow{a}\mid a\in \bigcup_{G\subseteq I}A_G\right\},V
\right\rangle,$
where $S$ is a set of states; $\sim_i$ is an equivalence relation on $S$ for each $i\in I$; $A_G$ is a set of atomic group actions for each $G\subseteq I$ such that: $G\subsetneq H$ implies $A_G\cap A_H=\emptyset$; $A_{\emptyset}=\emptyset$; $\xrightarrow{a}$ is a binary relation on $S$ for each $a\in \bigcup_{G\subseteq I}A_G$, called the transition relation of $a$; $V:\mathcal{P}\to 2^S$ is a valuation function. The family $\{A_G\}_{G\subseteq I}$ is part of the model, so the input size depends on how this family is represented. We assume an explicit representation in which only non-empty $A_G$'s are listed.

For each group $G\subseteq I$, the distributed indistinguishability relation is defined by $\sim_G := \bigcap_{i\in G}\sim_i$. We write $[s]_G$ for the $\sim_G$-equivalence class of $s$, and $[S]_G$ for the corresponding quotient set. We also write $A_G^+:=\bigcup_{H\subseteq G}A_H$ for the set of atomic actions available to subgroups of $G$.

Fix an enumeration $I=\{i_1,\ldots,i_n\}$. For every nonempty group $G=\{i_{n_1},\ldots,i_{n_k}\}$, let $\min G:=\min\{n_1,\ldots,n_k\}$. For mutually disjoint nonempty groups $G$ and $H$, define $G\prec H$ iff $\min G<\min H$. For each group $G$, its distributed action set is defined as $A_G^*:=A_G^+\cup\{\langle d_1,\ldots,d_n\rangle\in A_{G_1}^*\times\cdots\times A_{G_n}^*\mid \{G_1,\ldots,G_n\}$ is a non-trivial partial partition of $G$ and $G_1\prec\cdots\prec G_n\}$. Here, a non-trivial partial partition of a set $X$ is a partition of a not necessarily proper subset of $X$ that is not a singleton set. It follows that $A_\emptyset^*=\emptyset$ and $A_{\{i\}}^*=A_{\{i\}}$ for each $i\in I$.

For each distributed action $d=\langle d_1,\dots,d_n\rangle\in A_I^*$, the distributed transition relation is defined by $\xrightarrow{d}:=\bigcap_{1\le k\le n}\xrightarrow{d_k}$. A distributed action $d\in A_I^*$ is executable on a nonempty set $X\subseteq S$ if, for every $s\in X$, there is a state $t\in S$ such that $s\xrightarrow{d}t$. A strategy for a group $G$ is a partial function $\sigma_G:[S]_G\to A_G^*$ such that, whenever $[s]_G\in\mathrm{dom}(\sigma_G)$, the action $\sigma_G([s]_G)$ is executable on $[s]_G$. In particular, the empty
function is a strategy.

A possible execution of $\sigma_G$ starting from $[s]_G$ is a nonempty finite or infinite sequence $[s_1]_G[s_2]_G\cdots$ with $[s_1]_G=[s]_G$ such that, for every $j\geq 1$ for which $[s_{j+1}]_G$ occurs in the sequence, $[s_j]_G\xrightarrow{\sigma_G([s_j]_G)}[s_{j+1}]_G$. A possible execution is complete if it is infinite, or if it is finite and its last class does not belong to $\mathrm{dom}(\sigma_G)$. If a finite possible execution is of the form $[s_1]_G\cdots[s_n]_G$, then $[s_n]_G$ is its leaf-node and every $[s_j]_G$, with $1\le j<n$, is an inner-node. If it is infinite, then all its nodes are inner-nodes. For all complete executions starting from $[s]_G$, we denote their leaf-nodes and inner-nodes by $\mathrm{CELeaf}(\sigma_G,s)$ and $\mathrm{CEInner}(\sigma_G,s)$, respectively.

The semantics is defined as follows. Given a model $\mathcal{M}$, a state $s\in S$, and a formula $\varphi\in \mathsf{DKH}$, we omit the Boolean cases.

\begin{tabularx}{\textwidth}{lclX}

$\mathcal{M}, s \models K_G \varphi$ & \textit{iff} & & $\mathcal{M}, s' \models \varphi$ for all $s' \in [s]_G$ \\
$\mathcal{M}, s \models Kh_G \varphi$ & \textit{iff} & & there is a strategy $\sigma_G$ of $G$ such that: (1) $[t]_G \subseteq \llbracket \varphi \rrbracket^{\mathcal{M}}$ for all $[t]_G \in \mathrm{CELeaf}(\sigma_G,s)$, and (2) all complete executions of $\sigma_G$ starting from $[s]_G$ are finite. 
\end{tabularx}

\section{The Upper Bound}

The model checking procedure $\textsc{CheckerDKH}(\mathcal M,s,\varphi)$ (Algorithm~\ref{checkDKH}) decides whether the input state $s$ satisfies the formula $\varphi$ on a finite model $\mathcal M$. It is defined via an evaluation procedure $\textsc{Eval}(\mathcal M,\varphi)$, which computes the satisfaction set of $\varphi$ bottom-up, following the structure of formulas. 

Thus, $\textsc{CheckerDKH}(\mathcal M,s,\varphi)$ returns true iff $s\in \textsc{Eval}(\mathcal M,\varphi)$. The Boolean cases and the $K$-modality are treated as usual. The non-trivial case is the evaluation of a $Kh_G$-formula. For such a formula $Kh_G\psi$, Algorithm~\ref{algorithm for evalkh} first computes $\textsc{Eval}(\mathcal M,\psi)$ and then performs a monotone fixpoint construction over the quotient set $S/{\sim_G}$. Starting from the $\sim_G$-equivalence classes already contained in the  set, the algorithm repeatedly adds a class $E$ whenever there exists a distributed action that is executable on $E$ and whose possible successor classes are all contained in the current approximation.

The existence of such a witness is checked by an $\mathsf{NP}$ oracle. Algorithm~\ref{HasWinDistribution} gives the corresponding nondeterministic procedure. Note that an actual distributed action can be nested, but that it can always be flattened via Proposition 14 in \cite{DKH2025}. In particular, $\textsc{IsValidDAction}$ (Algorithm~\ref{isValidAction}) checks whether the guessed tuple can be viewed as an element of $A_G^*$, by assigning each action to a suitable subgroup and verifying that the chosen subgroups are defined and pairwise disjoint. Once the fixpoint is reached, $\textsc{EvalKh}$ returns the union of all states belonging to the winning equivalence classes.

\begin{algorithm}\label{checkDKH}
  \caption{{Model Checking Procedure}}
  \SetKwProg{Fn}{Procedure}{:}{}
  \Fn{\textsc{CheckerDKH}$(\mathcal{M},s,\phi)$}{
    \Return $s\in \textsc{Eval}(\mathcal{M},\phi)$
  }

  \SetKwProg{Fn}{Procedure}{:}{}
  
  \Fn{\textsc{Eval}{$(\mathcal{M}, \phi)$}}{
    \Switch{$\phi$}{
      \Case{$\phi \equiv \top$}{
        \Return $S$ \;
      }
      \Case{$\phi \equiv p$}{
        \Return $V(p)$ \;
      }
      \Case{$\phi \equiv \neg \psi$}{
        \Return $S \setminus \textsc{Eval}(\mathcal{M}, \psi)$ \;
      }
      \Case{$\phi \equiv \psi \wedge \chi$}{
        \Return $\textsc{Eval}(\mathcal{M}, \psi) \cap \textsc{Eval}(\mathcal{M}, \chi)$ \;
      }
      \Case{$\phi \equiv K_G \psi$}{
        \Return $\textsc{EvalK}(\mathcal{M}, G, \psi)$ \;
      }
      \Case{$\phi \equiv Kh_G \psi$}{
        \Return $\textsc{EvalKh}(\mathcal{M}, G, \psi)$ \;
      }
    }
  }
\end{algorithm}

\begin{algorithm}[H]
  \small
  \caption{{Evaluation of K-modality}}
  \SetKwProg{Fn}{Procedure}{:}{}
  \Fn{\textsc{EvalK}{$(\mathcal{M},G,\phi)$}}{
      $X \gets \textsc{Eval}(\mathcal{M},\phi)$ \;
      $\sim_G \gets \bigcap_{i \in G}\sim_i$ \;
      $\mathcal{C}_G \gets S/\sim_G$ \;
      $L \gets \emptyset$ \;
      \ForEach{$E \in \mathcal{C}_G$}{
          \If{$E \subseteq X$}{
            $L \gets L \cup E$ \;
          }
      }
      \Return $L$ \;
  }
\end{algorithm}

\begin{algorithm}\label{algorithm for evalkh}
\footnotesize
\caption{{Evaluation of Kh-modality}}
\SetKwProg{Fn}{Function}{:}{}
\Fn{\textsc{EvalKh}{$(\mathcal{M},G,\varphi)$}}{
    $X \gets \textsc{Eval}(\mathcal{M},\varphi)$ ;
    $\sim_G \gets \bigcap_{i\in G}\sim_i$ ;
    $\mathcal{C}_G \gets S/\sim_G$ ;
    $W \gets \{E \in \mathcal{C}_G \mid E \subseteq X\}$ \;
    $changed \gets \textnormal{True}$ \;
    \While{$changed$}{
        $changed \gets \textnormal{False}$ \;
        \ForEach{$E \in \mathcal{C}_G \setminus W$}{
            \If{$\textsc{HasWinDistribution}(\mathcal{M},G,E,W)$}{
                $W \gets W \cup \{E\}$ \;
                $changed \gets \textnormal{True}$ \;
            }
        }
    }
    \Return $\bigcup_{E\in W} E$\;
}
\end{algorithm}

\begin{algorithm}\label{HasWinDistribution}
\footnotesize
\caption{{A nondeterministic procedure}}
\SetKwProg{Fn}{Function}{:}{}
\Fn{\textsc{HasWinDistribution}{$(\mathcal{M},G,E,W)$}}{
    \textbf{nondeterministically guess}
    $(m,\langle (H_1,a_1),\dots,(H_m,a_m)\rangle)$,
    where $1\le m\le |G|$, $H_i\subseteq G$, and
    $a_i\in A_G^+$\;
    \If{not $\textsc{IsValidDAction}(\mathcal{M},G,\langle (H_1,a_1),\dots,(H_m,a_m)\rangle)$}
    {
        \Return \textnormal{False} \;
    }
    $\to^d \gets \bigcap_{k=1}^m \to^{a_k}$\;
    \If{$\bigl(\forall x\in E\,\exists y\in S\; x\to^{d}y\bigr)
      \land
      \bigl(\forall x\in E\,\forall y\in S\;(x\to^{d}y \Rightarrow [y]_G\in W)\bigr)$}{
    \Return \textnormal{True}\;
}
    \Return \textnormal{False} \;
}
\end{algorithm}

\begin{algorithm}\label{isValidAction}
\footnotesize
\caption{{Determine whether the guessed tuple represents a valid distributed action.}}
\SetKwProg{Fn}{Function}{:}{}
\Fn{\textsc{IsValidDAction}{$(\mathcal{M},G,\langle (H_1,a_1),\dots,(H_m,a_m)\rangle)$}}{
     \If{$m<1$}{
        \KwRet{\textnormal{False}}\;
    }

    \For{$i\gets 1$ \KwTo $m$}{
        \If{$H_i=\emptyset$ or $H_i\not\subseteq G$ or $a_i\notin A_{H_i}$}{
            \KwRet{\textnormal{False}}\;
        }
    }

    \For{$1\le i<j\le m$}{
        \If{$H_i\cap H_j\neq\emptyset$}{
            \KwRet{\textnormal{False}}\;
        }
    }

    \KwRet{\textnormal{True}}\;
}
\end{algorithm}

\subsection{Correctness}

It is easy to see that 

\medskip

\begin{lemma}\label{termination of hasWinDistribution}
  $\textsc{HasWinDistribution}$ always terminates.
\end{lemma}

Essentially, $\textsc{EvalKh}$ performs a smallest fixpoint computation. During the computation, it generates an increasing chain of sets of equivalence classes. Once the fixpoint is reached, the procedure terminates. The union of the final set of equivalence classes coincides with the set of states satisfying $Kh_G\psi$.

\medskip

\begin{lemma}\label{termination of EvalKh}
  $\textsc{EvalKh}$ always terminates. Moreover, for any input, there exists a sequence $W_0\subseteq W_1\subseteq \cdots \subseteq W_l$ of sets of $\sim_G$-equivalence classes such that $\bigcup W_l = \textsc{EvalKh}(\mathcal{M},G,\phi)$.
\end{lemma}

\begin{proof}
Let $\mathcal{C}_G=S/{\sim_G}$. Since $S$ is finite, $\mathcal{C}_G$ is finite as well. Each successful execution of the while-loop adds at least one new equivalence class to $W$, and no class is ever removed. Since $\mathcal{C}_G$ is finite, there can be at most $|\mathcal{C}_G|$ successful iterations. One final unsuccessful iteration sets $changed$ to $\textnormal{False}$, so the while-loop terminates.

Let $W_0=\{E\in \mathcal{C}_G \mid E\subseteq \textsc{Eval}(\mathcal{M},\phi)\}$ and, for each $i\ge 0$, let $W_{i+1}$ be the set obtained from $W_i$ after one execution of the while-loop in $\textsc{EvalKh}$. By construction, each $W_{i+1}$ is obtained from $W_i$ by adding some equivalence classes, so the sequence is increasing. When the algorithm terminates, it returns the union of all equivalence classes in the final set $W_l$. Therefore, $\textsc{EvalKh}(\mathcal{M},G,\phi)=\bigcup W_l.$
\end{proof}

\medskip

The next lemma shows the correctness of algorithm~\ref{isValidAction}.

\medskip

\begin{lemma}\label{hasValidDAction Correctness}
$\textsc{IsValidDAction}(\mathcal{M},G,\langle (H_1,a_1),\dots,(H_m,a_m)\rangle)$ returns True iff $\{H_1,\dots,H_m\}$ is a partial partition of $G$ and $a_i\in A_{H_i}$ for every $i\in[1,m]$. Consequently, if $m>1$, there is a permutation $\pi$, such that $\langle a_{\pi(1)},\dots,a_{\pi(m)}\rangle\in A_G^*$.
\end{lemma}

\begin{proof}
If the procedure returns True, then every $H_i$ is a nonempty subgroup of $G$, the $H_i$'s are pairwise disjoint, and $a_i\in A_{H_i}$ for each $i$. Thus, $\{H_1,\dots,H_m\}$ is a partial partition of $G$. If $m=1$, then $a_1\in A_{H_1}\subseteq A_G^+$, so $a_1\in A_G^*$.  If $m>1$, reordering these groups according to $\prec$ gives $H_{\pi(1)}\prec\cdots\prec H_{\pi(m)}$, and hence the corresponding tuple $\langle a_{\pi(1)},\dots,a_{\pi(m)}\rangle$ is an element of $A_G^*$.

Conversely, if $\{H_1,\dots,H_m\}$ is a partial partition of $G$ and $a_i\in A_{H_i}$ for every $i$, then all tests in the procedure succeed, so it returns True.
\end{proof}

\medskip

The next lemma shows the correctness of the algorithm~\ref{HasWinDistribution}.

\medskip

\begin{lemma}\label{hasWinDistribution Correctness}
$\textsc{HasWinDistribution}(\mathcal{M},G,E,W)$ returns $true$ iff there is a distributed action $c$ such that $c$ is executable on $E$ and applying $c$ to any $x\in E$ always leads to a state whose equivalence class is in $W$.
\end{lemma}

\begin{proof}
($\Rightarrow$): Assume that $\textsc{HasWinDistribution}(\mathcal{M},G,E,W)$ returns $true$. Then there is $m\le |G|$ and a tuple $\langle (H_1,a_1),\dots,(H_m,a_m)\rangle$ such that $\textsc{IsValidDAction}(\mathcal{M},G,\langle (H_1,a_1),\dots,(H_m,a_m)\rangle)$ returns True. By lemma~\ref{hasValidDAction Correctness}, $\{H_1,\dots,H_m\}$ is a partial partition of $G$ and $a_i\in A_{H_i}$ for every $i$. If $m=1$, let $c:=a_1$. If $m>1$, then there is a permutation $\pi$ such that $H_{\pi(1)}\prec\cdots\prec H_{\pi(m)}$ and $\langle a_{\pi(1)},\dots,a_{\pi(m)}\rangle\in A_G^*$. Let $c:=\langle a_{\pi(1)},\dots,a_{\pi(m)}\rangle$. In both cases, $\xrightarrow{c}=\bigcap_{j=1}^m\xrightarrow{a_j}$. Following the procedure, we have $\forall x\in E\ \exists y\in S\ (x\xrightarrow{c} y)$ and $\forall x\in E\ \forall y\in S\ (x\xrightarrow{c} y \Rightarrow [y]_G\in W)$. The first condition says that $c$ is executable on $E$, and the second says that applying $c$ to any $x\in E$ always leads to a state whose equivalence class is in $W$.

($\Leftarrow$): Suppose that there is a distributed action $c$ such that $c$ is executable on $E$ and applying $c$ to any $x\in E$ always leads to a state whose equivalence class is in $W$. Since $c\in A_G^*$, by Proposition 14 from \cite{DKH2025} we obtain actions $a_1\in A_{G_1},\dots,a_m\in A_{G_m}$ for some partial partition $P=\{G_1,\dots,G_m\}$ of $G$, such that $\xrightarrow{c}=\bigcap_{j=1}^m \xrightarrow{a_j}$. By lemma~\ref{hasValidDAction Correctness}, $\textsc{IsValidDAction}(\mathcal{M},G,\langle (G_1,a_1),\dots,(G_m,a_m)\rangle)$ returns True. By hypothesis, $\forall x\in E\ \exists y\in S\ (x\xrightarrow{c} y)$ and $\forall x\in E\ \forall y\in S\ (x\xrightarrow{c} y \Rightarrow [y]_G\in W)$ are all satisfied. Therefore, the nondeterministic branch of $\textsc{HasWinDistribution}$ that guesses $(m,\langle (G_1,a_1),\dots,(G_m,a_m)\rangle)$ leads to the output $true$.
\end{proof}

\medskip

The following lemma shows the correctness of the evaluation procedure.

\medskip

\begin{lemma}\label{Correctness lemma}
  $x\in \textsc{Eval}(\mathcal{M},\phi)$ iff $x\in \llbracket \phi \rrbracket$.
\end{lemma}
\begin{proof}
  Proceed by induction on the structure of $\phi$.

  \bigskip
  
  Boolean cases are immediate.

  \bigskip

  (IH1) For every proper subformula $\psi$ of $\phi$ and every state $x$, we have $x\in \textsc{Eval}(\mathcal{M},\psi)$ iff $x\in \llbracket \psi \rrbracket$. 

  \paragraph{($\phi\equiv K_G \psi$)}
  Assume $x\in \textsc{Eval}(\mathcal{M},\phi)$. Since $\phi\equiv K_G\psi$, this means that
  $x\in \textsc{EvalK}(\mathcal{M},G,\psi)$. By the definition of $\textsc{EvalK}$, we have
  $[x]_G\subseteq \textsc{Eval}(\mathcal{M},\psi)$. By IH1, $\textsc{Eval}(\mathcal{M},\psi)=\llbracket \psi \rrbracket$.
  Hence, $[x]_G\subseteq \llbracket \psi \rrbracket$. In other words, for all
  $s\in [x]_G$, we have $s\in \llbracket \psi \rrbracket$. By the semantics of
  $K_G$, we obtain $x\in \llbracket K_G\psi \rrbracket$.

  Conversely, assume $x\in \llbracket \phi \rrbracket$. Since $\phi\equiv K_G\psi$,
  we have $x\in \llbracket K_G\psi \rrbracket$. By the semantics of $K_G$, for all
  $s\in [x]_G$, we have $s\in \llbracket \psi \rrbracket$. By IH1,
  $\llbracket \psi \rrbracket = \textsc{Eval}(\mathcal{M},\psi)$. Thus, $[x]_G\subseteq \textsc{Eval}(\mathcal{M},\psi)$.
  By the definition of $\textsc{EvalK}$, this implies $x\in \textsc{EvalK}(\mathcal{M},G,\psi)$, and hence
  $x\in \textsc{Eval}(\mathcal{M},\phi)$.

  \paragraph{($\phi\equiv Kh_G \psi$)} 
  \begin{claim}
  $x\in \textsc{EvalKh}(\mathcal{M},G,\psi)$ iff there is a strategy $\sigma_G$ of $G$ that 1. for all $[t]_G\in \mathrm{CELeaf}(\sigma_G, x)$, $[t]_G\subseteq \llbracket \psi \rrbracket$, and 2. all its complete executions starting from $[x]_G$ are finite.
  \end{claim}
  \begin{proof}
  ($\Rightarrow$): Assume that $x\in \textsc{EvalKh}(\mathcal{M},G,\psi)$. By the lemma~\ref{termination of EvalKh}, $\textsc{EvalKh}$ always terminates. During the execution of $\textsc{EvalKh}$, record the order in which equivalence classes are added to $W$. Classes in the initial set $W_0$ have rank $0$. If a class $E$ is added later, define its rank to be the number of its position in this adding order. When $E$ is added, $\textsc{HasWinDistribution}(\mathcal{M},G,E,W')$ returns $true$ for the current set $W'$, and every class in $W'$ has smaller rank than $E$. By the lemma~\ref{hasWinDistribution Correctness}, there is a transition $c_E$ such that $c_E$ is executable on $E$ and applying $c_E$ to any $x\in E$ always leads to a state whose equivalence class is in $W'$.

  We define the strategy $\sigma_G$ by setting $\sigma_G(E)=c_E$ for every class $E$ added after the initial set, and leave $\sigma_G$ undefined on classes in $W_0$. Along every execution according to $\sigma_G$, the rank of the current class strictly decreases until a class in $W_0$ is reached. Hence, all complete executions are finite, and their leaf classes are contained in $W_0\subseteq \textsc{Eval}(\mathcal{M},\psi)=\llbracket\psi\rrbracket$ by IH1. Such $\sigma_G$ is the desired strategy.
  
  ($\Leftarrow$): Suppose that there exists a strategy $\sigma_G$ of $G$ such that (1) for all $[t]_G\in CELeaf(\sigma_G,x)$, we have $[t]_G\subseteq \llbracket \psi \rrbracket$, and (2) all complete executions of $\sigma_G$ starting from $[x]_G$ are finite.

  Let $\mathcal{C}_G=S/{\sim_G}$. Let $W_0\subseteq W_1\subseteq \cdots \subseteq W_l$ be the iteration sequence computed by $\textsc{EvalKh}$, where $W_0=\{E\in \mathcal{C}_G \mid E\subseteq \textsc{Eval}(\mathcal{M},\psi)\}$. For convenience, set $W_i=W_l$ for all $i>l$. We only consider those equivalence classes that are reachable from $[x]_G$ under $\sigma_G$. For every such reachable class $E$, define $d(E)$ to be the maximal length of a complete execution from $E$ to a leaf. This is well-defined and finite, since the set of reachable equivalence classes is finite and a reachable cycle would generate an infinite complete execution, contradicting the assumption that all complete executions of $\sigma_G$ starting from $[x]_G$ are finite.

  We show by induction on $d(E)$ that $E\in W_{d(E)}$ for every reachable class $E$. If $d(E)=0$, then $E$ is a leaf of a complete execution. By assumption, $E\subseteq \llbracket \psi \rrbracket$. By IH1, $\llbracket \psi \rrbracket = \textsc{Eval}(\mathcal{M},\psi)$. Hence, $E\subseteq \textsc{Eval}(\mathcal{M},\psi)$, and therefore $E\in W_0$.

  \bigskip

  (IH2) Assume that the statement holds for all reachable classes of depth at most $N$.

  \bigskip
  
  Let $E$ be a reachable class such that $d(E)=N+1$. Since $E$ is not a leaf, $\sigma_G(E)$ is defined. Let $c:=\sigma_G(E)$. Since $\sigma_G$ is a strategy, $c$ is executable on $E$. Take any $c$-successor class $E'$ of $E$. Then one step has already been taken, so every complete execution from $E'$ to a leaf has length at most $N$. Thus $d(E')\leq N$. By IH2, $E'\in W_N$. Since $E'$ was arbitrary, every $c$-successor class of $E$ belongs to $W_N$. Together with the fact that $c$ is executable on $E$, by the lemma~\ref{hasWinDistribution Correctness}, $\textsc{HasWinDistribution}(\mathcal{M},G,E,W_N)$ returns $true$. Hence, $E\in W_{N+1}$.

  Since $[x]_G$ is reachable to itself trivially, $[x]_G\in W_{d([x]_G)}$. Therefore, $x\in \bigcup W_l = \textsc{EvalKh}(\mathcal{M},G,\psi)$.
  \end{proof}
  By the claim, $x\in \textsc{EvalKh}(\mathcal{M},G,\psi)$ iff $x\in \llbracket Kh_G\psi \rrbracket$. Since $\phi\equiv Kh_G\psi$, we have $x\in \textsc{Eval}(\mathcal{M},\phi)$ iff $x\in \textsc{EvalKh}(\mathcal{M},G,\psi)$. Therefore, $x\in \textsc{Eval}(\mathcal{M},\phi)$ iff $x\in \llbracket \phi \rrbracket$.
\end{proof}

Finally, we are ready to show the correctness of the model checking procedure.

\begin{theorem}\label{Correctness}
  $\textsc{CheckerDKH}(\mathcal{M},x,\phi)$ returns True iff $\mathcal{M},x\models \phi$.
\end{theorem}
\begin{proof}
  By the definition of $\textsc{CheckerDKH}$, it returns True iff $x\in\textsc{Eval}(\mathcal M,\phi)$. By Lemma~\ref{Correctness lemma}, this is equivalent to $\mathcal M,x\models\phi$.
\end{proof}

\subsection{Complexity}

We define the $\textsc{HasWinDistribution}$ problem:

\begin{tcolorbox}[
colback=white,
colframe=black,
boxrule=0.8pt,
sharp corners,
left=8pt,
right=8pt,
top=7pt,
bottom=7pt
]
\underline{\textsc{HasWinDistribution}}

\vspace{4pt}
\textbf{Input:} a finite DKH model $\mathcal{M}$, a group $G\subseteq I$, an equivalence class $E\in S/{\sim_G}$, and a set $W\subseteq S/{\sim_G}$.

\textbf{Question:} does there exist a distributed action $c$ such that $c$ is executable on $E$, and for every $x\in E$ and every $y\in S$, if $x\xrightarrow{c}y$, then $[y]_G\in W$?
\end{tcolorbox}

The DKH model checking problem is defined as follows:
\begin{tcolorbox}[
colback=white,
colframe=black,
boxrule=0.8pt,
sharp corners,
left=8pt,
right=8pt,
top=7pt,
bottom=7pt
]
\underline{\textsc{CheckDKH}}

\vspace{4pt}
\textbf{Input:} a finite DKH model $\mathcal{M}$, a state $x\in S$, and a DKH formula $\phi$.

\textbf{Question:} does $\mathcal{M},x\models \phi$ hold?
\end{tcolorbox}

Take $(m,\langle (H_1,a_1),\dots,(H_m,a_m)\rangle)$ as a certificate, where $1\le m\le |G|$, $H_i\subseteq G$, and $a_i\in A_G^+$. We can show that

\begin{lemma}\label{HasWinDistribution in NP}
The problem \textsc{HasWinDistribution} is in $\mathsf{NP}$.
\end{lemma}

Finally, we obtain the upper bound for the problem.

\begin{theorem}\label{upper_bound}
The problem \textsc{CheckDKH} is in $\Delta_2^p$.
\end{theorem}

\begin{proof}
The procedure $\textsc{Eval}$ evaluates formulas bottom-up. Boolean cases and the $K_G$-case are handled by polynomial-time set operations over $S$ and the equivalence classes of $\sim_G$.

For a subformula $Kh_G\psi$, after computing $\textsc{Eval}(\mathcal M,\psi)$, the algorithm $\textsc{EvalKh}$ performs a monotone fixpoint computation over $S/{\sim_G}$. We can add at most $|S/{\sim_G}|\le |S|$ equivalence classes, and each iteration checks only polynomially many classes. Each such check is a call to \textsc{HasWinDistribution}, which is in $\mathsf{NP}$ by Lemma~\ref{HasWinDistribution in NP}. Thus each $Kh_G$-subformula is evaluated in deterministic polynomial time with an $\mathsf{NP}$ oracle.

Since the number of subformulas of $\phi$ is linear in $|\phi|$, the whole bottom-up evaluation runs in $P^{NP}$ (also called $\Delta_2^p$). By Theorem~\ref{Correctness}, $\textsc{CheckerDKH}(\mathcal M,x,\phi)$ returns true iff $\mathcal M,x\models\phi$. Therefore, \textsc{CheckDKH} is in $\Delta_2^p$.
\end{proof}

\section{The Lower Bound}
We show that the model checking problem for the logic of distributed knowing how is $\Delta_2^p$-hard by a reduction from a $\Delta_2^p$-complete problem SNSAT \cite{SNSAT}.

An instance $\mathcal I$ of SNSAT is given by a set $X=\{x_1,\ldots,x_n\}$ of boolean variables together with a list $\mathcal L$ of equivalences
\[
\begin{array}{rcl}
x_1 &\Leftrightarrow& \exists Z_1\ F_1(Z_1),\\
x_2 &\Leftrightarrow& \exists Z_2\ F_2(x_1,Z_2),\\
    &\vdots& \\
x_n &\Leftrightarrow& \exists Z_n\ F_n(x_1,\ldots,x_{n-1},Z_n),
\end{array}
\]
where, for $i=1,\ldots,n$, $Z_i=\{z_1^i,\ldots,z_{p_i}^i\}$ is a set of boolean variables, and $F_i$ is a boolean formula with variables among $Z_i\cup\{x_1,\ldots,x_{i-1}\}$. Note that in $\mathcal I$ the sets $X,Z_1,\ldots,$ and $Z_n$ are pairwise
disjoint. We write $Z=\bigcup_{i=1}^n Z_i$ and $Var=X\cup Z$.

The equivalences $\mathcal L$ in $\mathcal I$ define a unique valuation
$v_{\mathcal I}$ of the variables in $X$:
\[
  v_{\mathcal I}(x_i)=\top
  \overset{\mathrm{def}}{\Longleftrightarrow}
  F_i(v_{\mathcal I}(x_1),\ldots,v_{\mathcal I}(x_{i-1}),Z_i)
  \text{ is satisfiable.}
\]
Observe that there exists a simple algorithm in $\Delta^p_2$ that computes $v_{\mathcal I}$ one value at a time. When $v_{\mathcal I}$ is known over $\{x_1,\ldots,x_{i-1}\}$, the value of $v_{\mathcal I}(x_i)$ is computed by solving a boolean satisfiability problem, ``is $F_i$ satisfiable with the given values of $x_1,\ldots,x_{i-1}$?'', for which a SAT oracle is sufficient.

\begin{tcolorbox}[
colback=white,
colframe=black,
boxrule=0.8pt,
sharp corners,
left=8pt,
right=8pt,
top=7pt,
bottom=7pt
]
\underline{\textsc{SNSAT}}

\vspace{4pt}
\textbf{Input:} an instance $\mathcal I$ as above.

\textbf{Question:} decide whether $v_{\mathcal I}(x_n)=\top$?
\end{tcolorbox}
An \textsc{SNSAT} instance $\mathcal I$ is positive if $v_{\mathcal I}(x_n)=\top$.

\subsection{The reduction}

In the following reduction, we do not specify the order of the components in a distributed action. Since this order has no semantic effect, every tuple is understood to be rearranged according to the fixed order $\prec$ when necessary.

Keep the notions above. Let $\mathcal I$ be an instance of $\textsc{SNSAT}$. W.L.O.G., suppose each $F_i$ is in CNF. Write $Z_i=\{z_1^i,\ldots,z_{p_i}^i\}$ and $F_i=C^i_1\wedge\cdots\wedge C^i_{m_i}$. We also assume that every previous variable $x_j$ with $j<i$ occurs in $F_i$. If some previous variable $x_j$ does not occur in $F_i$, we conjoin the tautological clause $(x_j\lor\neg x_j)$ to $F_i$. We also assume that every $F_i$ contains at least one variable. If $F_i$ contains no variables, introduce a fresh local variable $z_i^\ast$, add it to $Z_i$, and conjoin the tautological clause $(z_i^\ast\lor\neg z_i^\ast)$ to $F_i$. We build a reduction from it to an instance of $\textsc{CheckDKH}$ as follows.

\paragraph{Model}
Let $\mathcal{M}=\left\langle S,\{\sim_a\}_{a\in I},\{A_G\}_{G\subseteq I},\left\{\xrightarrow{a}\mid a\in \bigcup_{G\subseteq I}A_G\right\},V\right\rangle$.
\begin{itemize}
  \item The set of states is defined as follows. First, create a success state $s^{\top}$. For each layer $i\in[1,n]$, create an entry state $s^i$. For each clause $C^i_r$ of $F_i$, create a bad state $b_r^i$. Finally, for every  layer $k$ and every previous variable $x_j$ with $1\le j<k\le n$, create a negative-check state $n^k_{x_j}$. Thus, $S=\{s^i:1\le i\le n\}\cup\{s^\top\}\cup\{b^i_r:1\le i\le n,\ 1\le r\le m_i\}\cup\{n^k_{x_j}:1\le j<k\le n\}$.

  \item Now we define the set of agents $I$. For every  layer $k$, the group $G_k$ contains $k$-indexed copies of all agents needed to evaluate layers $i\le k$. More precisely, for every layer $i\le k$ and every local variable $z_m^i\in Z_i$, we introduce an agent $a^k_{z_m^i}$. For every layer $i\le k$ and every previous variable occurrence $x_j$ in $F_i$, with $1\le j<i$, we introduce an agent $a^k_{x_j^i}$. Hence, $G_k=\{a^k_{z_m^i}:1\le i\le k,\ 1\le m\le p_i\}\cup\{a^k_{x_j^i}:1\le j<i\le k\}$, and the full set of agents is $I=\bigcup_{k=1}^n G_k$. Agents with different  superscripts are distinct.

  \item All epistemic relations are identity relations: $s\sim_a t\iff s=t$. So the equivalence classes are singleton sets.

  \item For every layer $k$, every layer $i\le k$, and every local variable $z_m^i\in Z_i$, set $A_{\{a^k_{z_m^i}\}}=\{t^k_{z_m^i},f^k_{z_m^i}\}$. For every layer $k$, every layer $i\le k$, and every previous variable occurrence $x_j$ in $F_i$, with $1\le j<i$, set $A_{\{a^k_{x_j^i}\}}=\{t^k_{x_j^i},f^k_{x_j^i}\}$. For every non-singleton group $H$, set $A_H=\varnothing$. 
  
  \item We define the transition relation from entry states layer by layer. For a state $s$ and an atomic action $a$, write $Img(s,a):=\{t\in S\mid s\xrightarrow{a}t\}$. For any $i\le k$, and for every action not specified below, the image set is empty. For each local variable $z_m^i\in Z_i$, define $Img(s^i,t^k_{z_m^i})=\{s^\top\}\cup\{b^i_r:z_m^i\text{ does not occur in }C^i_r\}\cup\{s^j,n^k_{x_j}:1\le j<i\}$ and $Img(s^i,f^k_{z_m^i})=\{s^\top\}\cup\{b^i_r:\neg z_m^i\text{ does not occur in }C^i_r\}\cup\{s^j,n^k_{x_j}:1\le j<i\}$. For each previous variable occurrence $x_j^i$ with $1\le j<i$, define $Img(s^i,t^k_{x_j^i})=\{s^\top,s^j\}\cup\{b^i_r:x_j^i\text{ does not occur in }C^i_r\}\cup\{s^h,n^k_{x_h}:1\le h<i,\ h\ne j\}$ and $Img(s^i,f^k_{x_j^i})=\{s^\top,n^k_{x_j}\}\cup\{b^i_r:\neg x_j^i\text{ does not occur in }C^i_r\}\cup\{s^h,n^k_{x_h}:1\le h<i,\ h\ne j\}$.

  \item For every negative-check state $n^k_{x_j}$, the state is a leaf for the current  group $G_k$: for every atomic action $a$ of an agent in $G_k$, set $Img(n^k_{x_j},a)=\varnothing$. For every lower  layer $\ell<k$ and every atomic action $a$ of an agent in $G_\ell$, set $Img(n^k_{x_j},a)=Img(s^j,a)$. Finally, $s^\top$ and all bad states are leaves, so $Img(s^\top,a)=\varnothing$ and $Img(b^i_r,a)=\varnothing$ for every atomic action $a$.

  \item We use proposition letters $q,r_1,\ldots,r_n$. The proposition $q$ marks the success state, so $V(q)=\{s^\top\}$. For each  layer $k$, the proposition $r_k$ marks the negative-check states belonging to  $k$, so $V(r_k)=\{n^k_{x_j}:1\le j<k\}$. All other proposition letters are false everywhere.
\end{itemize}

\paragraph{Formula}
Define $\Theta_1:=q$ and $\Psi_1:=Kh_{G_1}q$. For $k>1$, define $\Theta_k:=q\vee(r_k\wedge\neg\Psi_{k-1})$ and $\Psi_k:=Kh_{G_k}\Theta_k$.

\paragraph{The CheckDKH Instance}
$\langle \mathcal{M}, s^n,\Psi_n\rangle$

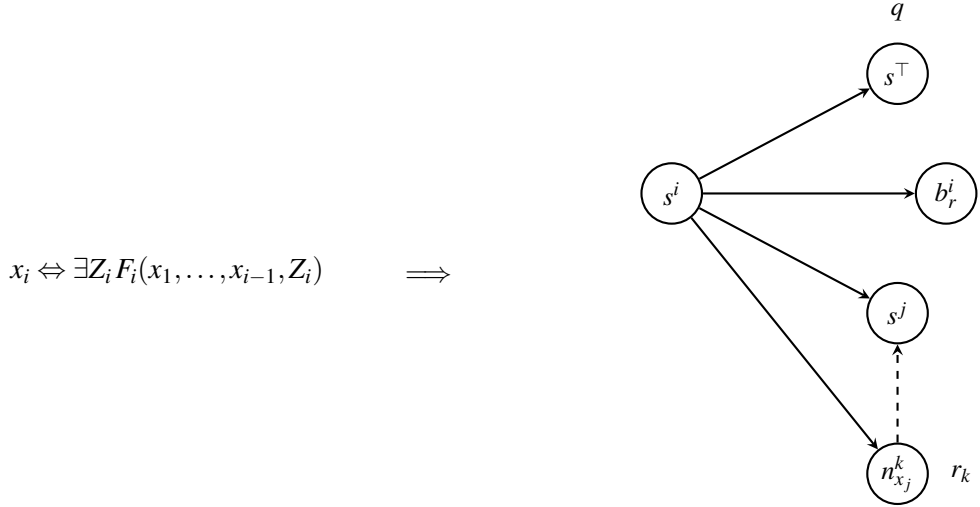
\begin{figure}
\centering

\begin{minipage}[c]{0.34\textwidth}
\centering
\[
x_i \Leftrightarrow \exists Z_i\,
F_i(x_1,\ldots,x_{i-1},Z_i)
\]
\end{minipage}
\begin{minipage}[c]{0.08\textwidth}
\centering
\[
\Longrightarrow
\]
\end{minipage}
\begin{minipage}[c]{0.54\textwidth}
\centering
\begin{tikzpicture}[
  >=stealth,
  thick,
  node distance=1.6cm,
  state/.style={circle,draw,minimum size=8mm,inner sep=1pt},
  bad/.style={circle,draw,minimum size=8mm,inner sep=1pt},
  every node/.style={font=\small}
]

\node[state] (si) {$s^i$};

\node[state, above right=1.0cm and 2.4cm of si] (top) {$s^\top$};
\node[bad, right=2.8cm of si] (bad) {$b_r^i$};
\node[state, below right=1.0cm and 2.4cm of si] (sj) {$s^j$};
\node[state, below=1.3cm of sj] (neg) {$n^k_{x_j}$};

\draw[->] (si) -- (top);
\draw[->] (si) -- (bad);
\draw[->] (si) -- (sj);
\draw[->] (si) -- (neg);

\draw[->,dashed] (neg) -- (sj);

\node[above=0.15cm of top] {$q$};
\node[right=0.15cm of neg] {$r_k$};

\end{tikzpicture}
\end{minipage}

\caption{A schematic view of one layer of the reduction.}
\label{fig:reduction-layer}
\end{figure}

\bigskip

It is easy to see that the construction is polynomial in $|\mathcal I|$.

\subsection{Examples}
Consider the following $\textsc{SNSAT}$ instance, where $x_1 \Leftrightarrow \exists z_1^1\, z_1^1$, and $x_2 \Leftrightarrow \exists z_1^2\,(x_1\wedge z_1^2)$. Then the reduction is as Figure~\ref{fig:reduction1}. Note that we omit the transitions from $n_{x_1}^2$ since it simply simulates $s^1$. It is easy to see that $(x_1\wedge z_1^2)$ is satisfiable by setting $z_1^2=\top$, given that $v_{\mathcal I}(x_1)=\top$. In the output of the reduction, we can take the strategy such that
\[
\sigma([s^2]_{G_2})=\langle t^2_{x_1^2},t^2_{z_1^2}\rangle
\quad\text{and}\quad
\sigma([s^1]_{G_2})= t^2_{z_1^1}.
\]
Then the possible successors from $[s^2]_{G_2}$ are either $[s^\top]_{G_2}$ or $[s^1]_{G_2}$, and from $[s^1]_{G_2}$ the only possible successor is $[s^\top]_{G_2}$. Hence, every complete execution is finite and ends in a state satisfying $q$. Therefore, $\mathcal M,s^2\models Kh_{G_2}\bigl(q\vee(r_2\wedge\neg Kh_{G_1}q)\bigr)$.

\begin{figure}[h]
\centering
\includegraphics[width=0.9\textwidth]{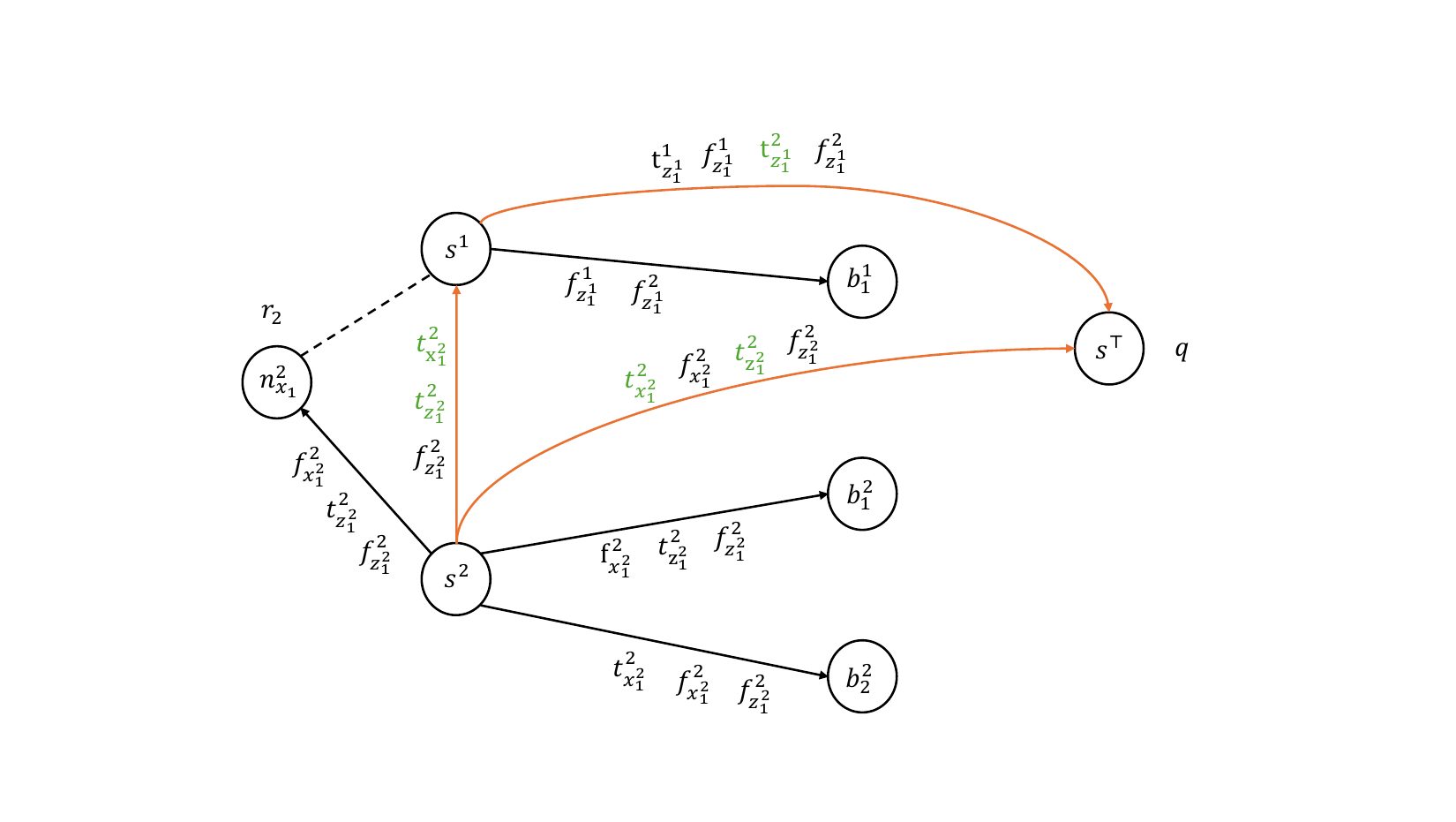}
\caption{The first reduction.}
\label{fig:reduction1}
\end{figure}

Now we modify the above $\textsc{SNSAT}$ instance by replacing $x_1$ with $\neg x_1$ in $(x_1\wedge z_1^2)$. The reduction is as Figure~\ref{fig:reduction2}. In this case, satisfying $F_2$ would require setting $x_1$ to false, but the previous variable $x_1$ has already been determined to be true. Thus $x_2$ is false. Correspondingly, every strategy for $G_2$ either reaches a bad state or reaches the negative-check state $n_{x_1}^2$, which witnesses that the false claim about $x_1$ is incompatible with the lower layer. Hence, $\mathcal M,s^2\not\models Kh_{G_2}\bigl(q\vee(r_2\wedge\neg Kh_{G_1}q)\bigr)$.

\begin{figure}[h]
\centering
\includegraphics[width=0.9\textwidth]{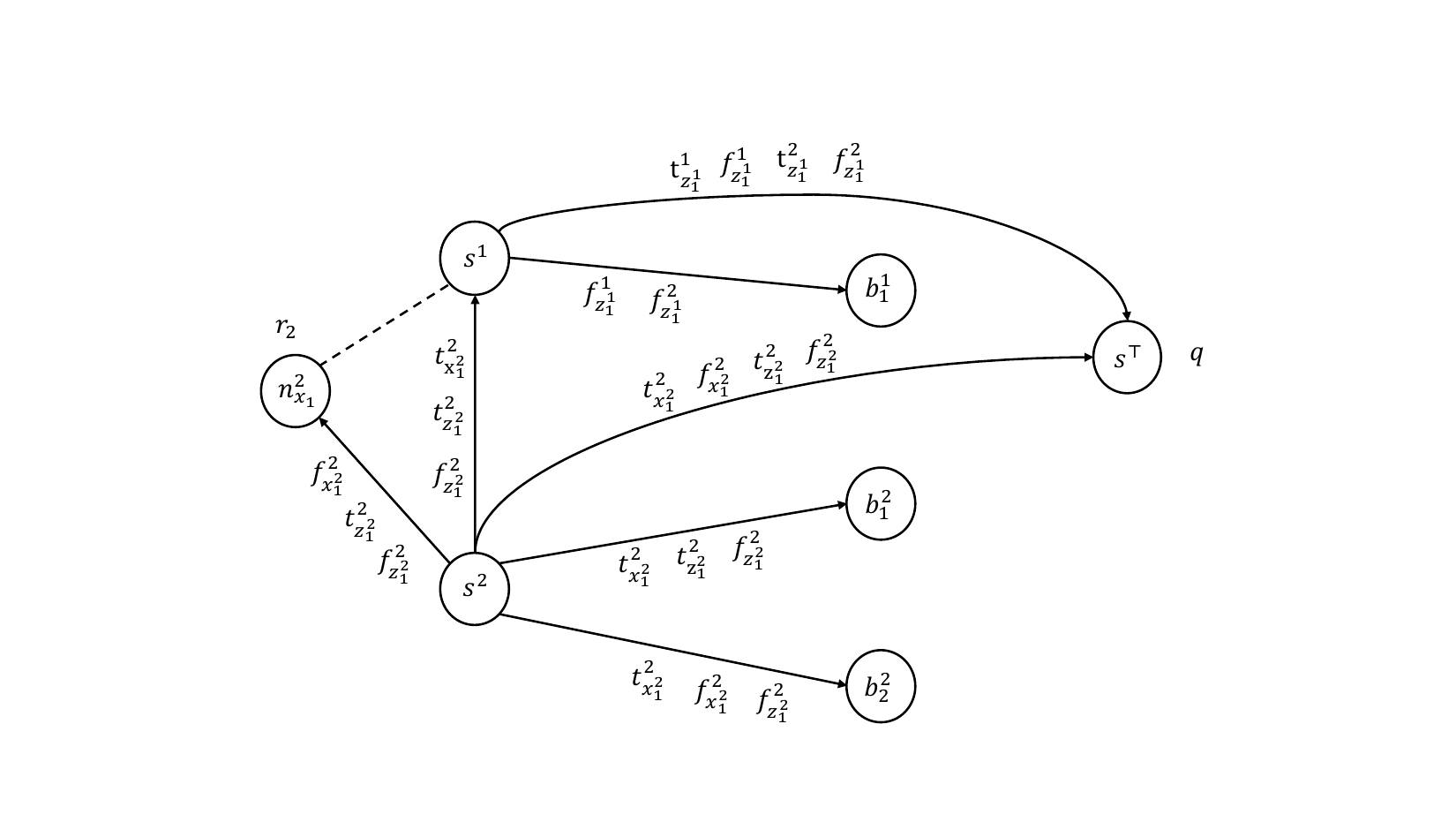}
\caption{The second reduction.}
\label{fig:reduction2}
\end{figure}

\subsection{Correctness}

Note that $d\in A_G^*$ can be nested in general, namely, $d$ could be of the form $\langle \langle a, b\rangle,c\rangle$. But in our case, things can be much simpler. The following lemma shows that every distributed action following the reduction should be flattened.

\medskip

\begin{lemma}\label{flatten}

Let $k$ be any layer. Let $d\in A^*_{G_k}$ be a distributed action that uses all and only singleton actions. Then there are distinct agents $a_1,\ldots,a_m\in G_k$ and actions $d_i\in A_{{a_i}}$ such that, if $m=1$, then $d=d_1$, and if $m>1$, then $d$ is of the form $\langle d_1,\ldots,d_m\rangle$.

\end{lemma}

\begin{proof}
  By Proposition~14 in \cite{DKH2025} together with the fact that all non-singleton groups have empty atomic action set.
\end{proof}

\medskip

The following lemma shows that every assignment for $F_i$ corresponds to a distributed action.

\medskip

\begin{lemma}\label{corresponding}
Let $k$ be any layer, and let $a_1,\ldots,a_m$ be exactly the agents in $G_k$ corresponding to the variables occurring in $F_k$. For any distributed action
\[
d=
\begin{cases}
d_1\in A_{\{a_1\}}, & \text{if } m=1,\\
\langle d_1,\ldots,d_m\rangle\in
A_{\{a_1\}}\times\cdots\times A_{\{a_m\}}\subseteq A^*_{G_k},
& \text{if } m>1.
\end{cases}
\]
$d$ uniquely determines a truth assignment $\alpha_k^d$ for $F_k$: if $d_\ell$ is the true action of $a_\ell$, then the corresponding variable is assigned $\top$; if $d_\ell$ is the false action of $a_\ell$, then the corresponding variable is assigned $\bot$. Conversely, every truth assignment for $F_k$ is induced by such a distributed action $d$. The statement holds for any $i\le k$.
\end{lemma}

\begin{proof}
Immediate from the definitions of $G_k$ and of the singleton action sets.
\end{proof}

\medskip

From now on, we treat Lemma~\ref{flatten} and Lemma~\ref{corresponding} as facts. The next lemma states the core mechanism of our reduction. If one picks an assignment that does not satisfy $F_i$, the corresponding distributed transition will lead to a bad state.

\medskip

\begin{lemma}\label{bad-clause-sat}
Let $i\le k$ be layers. Let $d\in A^*_{G_k}$ be a distributed action using exactly the singleton agents corresponding to the variables occurring in $F_i$, and let $\alpha_i^d$ be the assignment to the variables occurring in $F_i$ induced by $d$. Then, for every clause $C_r^i$ of $F_i$, $[s^i]_{G_k}\overset{d}{\to}[b_r^i]_{G_k}$ iff $C_r^i$ is false under $\alpha_i^d$. Consequently, there is no bad state $b_r^i$ such that $[s^i]_{G_k}\overset{d}{\to}[b_r^i]_{G_k}$ iff $\alpha_i^d$ satisfies $F_i$.
\end{lemma}

\begin{proof}
($\Rightarrow$): Fix a clause $C_r^i$ of $F_i$ and suppose that $[s^i]_{G_k}\overset{d}{\to}[b_r^i]_{G_k}$. By the definition of $\overset{d}{\to}$, for every atomic action $a$ occurring in $d$, we have $s^i\overset{a}{\to}b_r^i$. 

Observation: Let $a$ be an atomic action corresponding to a variable $\ell$. If $a=t_\ell^k$, then the selected literal is $\ell$. By definition, $s^i\overset{a}{\to} b_r^i$ iff $\ell$ does not occur in $C_r^i$. In this case, either $\neg\ell$ occurs in $C_r^i$ or it does not. In the former case, $\neg\ell$ is evaluated as $\bot$ under the assignment induced by $d$. In the latter case, no literal of the variable $\ell$ contributes $\top$ to $C_r^i$. On the other hand, if $a=f_\ell^k$, then the selected literal is $\neg\ell$. By definition, $s^i\overset{a}{\to} b_r^i$ iff $\neg\ell$ does not occur in $C_r^i$. In this case, either $\ell$ occurs in $C_r^i$ or it does not. In the former case, $\ell$ is evaluated as $\bot$ under the assignment induced by $d$. In the latter case, no literal of the variable $\ell$ contributes $\top$ to $C_r^i$.

By the observation, whether $a$ is a true action or a false action, the literal selected by $a$ contributes $\bot$ to $C_r^i$ under $\alpha_i^d$. Hence, no literal of $C_r^i$ is evaluated as $\top$ under $\alpha_i^d$. Therefore, $C_r^i$ is false under $\alpha_i^d$.

($\Leftarrow$): Prove by contraposition. If $[s^i]_{G_k}\not\overset{d}{\to}[b_r^i]_{G_k}$, then by the definition of $\overset{d}{\to}$, there is an atomic action $a$ in $d$, s.t. $s^i\not\overset{a}{\to}b_r^i$. If $a=t_{\ell}^k$, by the definition of $\overset{a}{\to}$, $\ell$ occurs in $C_r^i$. Then $\ell$ is evaluated as $\top$ under $\alpha_i^d$. Similarly, if $a=f_{\ell}^k$, then $\neg\ell$ occurs in $C_r^i$. Thus, $\neg\ell$ is evaluated as $\top$ under $\alpha_i^d$. In both cases, the selected literal is evaluated as $\top$ under $\alpha_i^d$. Therefore, $C_r^i$ is true under $\alpha_i^d$.

\end{proof}

\medskip

Intuitively, $s^j$ corresponds to the case when $x_j$ is claimed to be true, while $n^k_{x_j}$ corresponds to the case when $x_j$ is claimed to be false.

\medskip

\begin{lemma}\label{negative-state}
  For every pair of layers $j,k\in[1,n]$ such that $j<k$, we have $\mathcal M,n^k_{x_j}\models\Theta_k$ iff $\mathcal M,n^k_{x_j}\models\neg\Psi_{k-1}$. As a consequence, $\mathcal{M},n_{x_j}^k\models \Theta_k$ iff $\mathcal{M},s^j\models \neg\Psi_{k-1}$.
\end{lemma}

\begin{proof}
($\Rightarrow$): Suppose $\mathcal M,n^k_{x_j}\models\Theta_k$. Unfold the definition of $\Theta_k$, we infer that $\mathcal M,n^k_{x_j}\models q$ or $\mathcal M,n^k_{x_j}\models r_{k}\wedge \neg \Psi_{k-1}$. Since $n^k_{x_j}\not \in V(q)$ and $n^k_{x_j}\in V(r_k)$, we conclude that $\mathcal M,n^k_{x_j}\models \neg \Psi_{k-1}$. Since $Img(n^k_{x_j},a)=Img(s^j,a)$ holds for every atomic action $a$ of every group $G_\ell$ with $\ell<k$, and since all epistemic relations are identity relations, the two states have the same complete executions under every strategy for every group occurring in $\Psi_{k-1}$. Hence, $\mathcal{M},n^k_{x_j}\models\Psi_{k-1}$ iff $\mathcal{M},s^j\models\Psi_{k-1}$, and therefore $\mathcal{M},s^j\models\neg\Psi_{k-1}$.

($\Leftarrow$): Same reasoning as above, but in reverse direction.

\end{proof}

\medskip

We now move on to the correctness lemma for the reduction.

\medskip

\begin{lemma} \label{reduction_correctness}
  Let $\mathcal{I}$ be a SNSAT instance as above. Let $\langle \mathcal{M},s^n,\Psi_n\rangle$ be the output of the reduction. Then for every $k\in[1,n]$ and every $i\in[1,k]$, we have $v_{\mathcal I}(x_i)=\top$ iff $\mathcal M,s^i\models\Psi_k$. In particular, for every $i\in[1,n]$, we have $v_{\mathcal I}(x_i)=\top$ iff $\mathcal M,s^i\models\Psi_i$. 
\end{lemma}

\begin{proof}
  Proceed by induction on $k$, which is the index of $\Psi_k$. For each fixed $k$, we prove the claim for all $i\le k$ by induction on $i$.

($k=1$): Suppose $v_{\mathcal{I}}(x_1)=\top$. By the definition of SNSAT, there is an assignment $\alpha$ for the variables occurring in $F_1$ such that $F_1$ is true under $\alpha$. Let $d$ be the distributed action induced by $\alpha$. By Lemma~\ref{bad-clause-sat}, there is no bad state $b_r^1$ such that $[s^1]_{G_1}\overset{d}{\to}[b_r^1]_{G_1}$. Hence, every complete execution generated by the strategy $\sigma$ with $\sigma([s^1]_{G_1})=d$ reaches only $[s^\top]_{G_1}$. Since $\mathcal M,s^\top\models q$, we get $\mathcal M,s^1\models Kh_{G_1}q$. In other words, $\mathcal M,s^1\models\Psi_1$.

Conversely, suppose $\mathcal M,s^1\models\Psi_1$. Since $\Psi_1=Kh_{G_1}q$, there is a strategy $\sigma$ such that every complete execution from $[s^1]_{G_1}$ according to $\sigma$ ends in a class satisfying $q$. Let $d=\sigma([s^1]_{G_1})$. We may assume that $d$ uses exactly the singleton agents corresponding to the variables occurring in $F_1$, since adding missing relevant singleton actions only shrinks the successor set and preserves success. Since $s^\top$ is the only state satisfying $q$, no bad state $b_r^1$ can be reached from $[s^1]_{G_1}$ by $d$. Hence, by Lemma~\ref{bad-clause-sat}, the assignment $\alpha_1^d$ satisfies $F_1$. Therefore, by the definition of SNSAT, $v_{\mathcal I}(x_1)=\top$.

Now let $k>1$. Assume as outer induction hypothesis that for every $i\in [1,k)$, 
\[
v_{\mathcal I}(x_i)=\top \iff \mathcal M,s^i\models\Psi_{k-1}.
\]

We need to prove for every $i\in [1,k]$,

\[
v_{\mathcal I}(x_i)=\top \iff \mathcal M,s^i\models\Psi_{k}.
\]

 Fix a $k$, we then prove the statement by induction on $i$. The inner induction hypothesis is that for every $h\in [1,i)$, where $i\le k$,
\[
v_{\mathcal I}(x_h)=\top \iff \mathcal M,s^h\models\Psi_k.
\]

We prove that $v_{\mathcal I}(x_i)=\top$ iff $\mathcal M,s^i\models\Psi_k$.

($\Rightarrow$): Suppose $v_{\mathcal I}(x_i)=\top$. Then by definition of SNSAT, $\exists Z_i:F_i(x_1,\ldots,x_{i-1},Z_i)$. We build the strategy $\sigma$ as follows. For every $h\le i$ such that $v_{\mathcal I}(x_h)=\top$, choose an assignment $\alpha_h$ such that $\alpha_h$ agrees with $v_{\mathcal I}$ on $x_1,\ldots,x_{h-1}$ and satisfies $F_h$. Let $d_h\in A^*_{G_k}$ be the distributed action induced by $\alpha_h$ according to Lemma~\ref{corresponding}, and set $\sigma([s^h]_{G_k})=d_h$. Leave $\sigma$ undefined elsewhere.

In particular, since $v_{\mathcal I}(x_i)=\top$, the strategy is defined at $[s^i]_{G_k}$. Let $d_i=\sigma([s^i]_{G_k})$. By Lemma~\ref{bad-clause-sat}, no bad state $b_r^i$ is reachable from $[s^i]_{G_k}$ by $d_i$. More generally, for every $h\le i$ such that $\sigma([s^h]_{G_k})$ is defined, no bad state $b_r^h$ is reachable from $[s^h]_{G_k}$ by $\sigma([s^h]_{G_k})$.

Now consider any one-step successor of a class $[s^h]_{G_k}$ on which $\sigma$ is defined. There are three possible cases:
\begin{itemize}
  \item If $[s^h]_{G_k}\overset{\sigma([s^h]_{G_k})}{\to}[s^{\top}]_{G_k}$, since $s^{\top}\in V(q)$, then $\mathcal{M}, s^{\top}\models \Theta_k$.
  \item If $[s^h]_{G_k}\overset{\sigma([s^h]_{G_k})}{\to}[n_{x_j}^k]_{G_k}$ for some $j<h$, then $\sigma([s^h]_{G_k})$ assigns $x_j^h$ to false. Hence, $v_{\mathcal I}(x_j)=\bot$. By the outer induction hypothesis, $\mathcal M,s^j\not\models\Psi_{k-1}$. By Lemma~\ref{negative-state}, $\mathcal M,n_{x_j}^k\models\Theta_k$.
  \item If $[s^h]_{G_k}\overset{\sigma([s^h]_{G_k})}{\to}[s^j]_{G_k}$ for some $j<h$, then $\sigma([s^h]_{G_k})$ assigns $x_j^h$ to true. Hence, $v_{\mathcal I}(x_j)=\top$, so $\sigma$ is defined at $[s^j]_{G_k}$.
\end{itemize}
Thus, along every complete execution from $[s^i]_{G_k}$ according to $\sigma$, the layer of entry states strictly decreases whenever the execution moves to another entry state. Hence, all complete executions are finite. Their leaves are either $[s^\top]_{G_k}$ or some $[n_{x_j}^k]_{G_k}$ satisfying $\Theta_k$. Therefore, $\mathcal M,s^i\models Kh_{G_k}\Theta_k$, that is, $\mathcal M,s^i\models\Psi_k$.

($\Leftarrow$): Suppose $\mathcal M,s^i\models\Psi_k$. Since $\Psi_k=Kh_{G_k}\Theta_k$, there is a $G_k$-strategy $\sigma$ such that every complete execution from $[s^i]_{G_k}$ according to $\sigma$ has a leaf satisfying $\Theta_k$. Let $d=\sigma([s^i]_{G_k})$. We may assume that $d$ uses exactly the singleton agents corresponding to the variables occurring in $F_i$, since adding missing relevant singleton actions only shrinks the successor set and preserves success. We first show that the assignment $\alpha_i^d$ satisfies $F_i$. No bad state $b_r^i$ is reachable from $[s^i]_{G_k}$ by $d$. This is because, if $[s^i]_{G_k}\overset{d}{\to}[b_r^i]_{G_k}$ for some $r$, then the one-step execution ending in $[b_r^i]_{G_k}$ would be a complete execution whose leaf does not satisfy $\Theta_k$, a contradiction. Hence, by Lemma~\ref{bad-clause-sat}, the assignment $\alpha_i^d$ satisfies $F_i$.

It remains to show that $\alpha_i^d$ agrees with $v_{\mathcal I}$ on the previous variables $x_1,\ldots,x_{i-1}$. Let $j<i$. If $d$ assigns $x_j^i$ to false, then by the definition of the transition relation at $s^i$, we have $[s^i]_{G_k}\overset{d}{\to}[n^k_{x_j}]_{G_k}$. Since every complete execution according to $\sigma$ has a leaf satisfying $\Theta_k$, we have $\mathcal M,n^k_{x_j}\models\Theta_k$. By Lemma~\ref{negative-state}, $\mathcal M,s^j\not\models\Psi_{k-1}$. By the outer induction hypothesis, $v_{\mathcal I}(x_j)=\bot$. If $d$ assigns $x_j^i$ to true, then by the definition of the transition relation at $s^i$, we have $[s^i]_{G_k}\overset{d}{\to}[s^j]_{G_k}$. Since every complete execution from $[s^i]_{G_k}$ according to $\sigma$ is successful, the continuation of $\sigma$ from $[s^j]_{G_k}$ is also successful, so $\mathcal M,s^j\models\Psi_k$. By the inner induction hypothesis, $v_{\mathcal I}(x_j)=\top$.

Therefore, $\alpha_i^d$ agrees with $v_{\mathcal I}$ on all previous variables $x_1,\ldots,x_{i-1}$. Since $\alpha_i^d$ satisfies $F_i$, it gives an assignment to $Z_i$ witnessing
$\exists Z_i\,F_i(x_1,\ldots,x_{i-1},Z_i)$. Thus, by the definition of SNSAT, $v_{\mathcal I}(x_i)=\top$.
\end{proof}

\medskip

Following Lemma~\ref{reduction_correctness}, we deduce

\medskip

\begin{theorem}\label{lower_bound}
  The problem \textsc{CheckDKH} is $\Delta_2^p$-hard.
\end{theorem}

\medskip

Finally, combining Theorem~\ref{upper_bound} and Theorem~\ref{lower_bound}, we conclude

\medskip

\begin{corollary}
  The problem \textsc{CheckDKH} is $\Delta_2^p$-complete.
\end{corollary}

\section{Conclusion and Future Work}
We show the $\Delta_2^p$-completeness for the model checking problem of the logic of distributed knowing how. A natural next step is to investigate the satisfiability problem for this logic.

\bibliographystyle{eptcs}

\begin{thebibliography}{1}
\providecommand{\bibitemdeclare}[2]{}
\providecommand{\surnamestart}{}
\providecommand{\surnameend}{}
\providecommand{\urlprefix}{Available at }
\providecommand{\url}[1]{\texttt{#1}}
\providecommand{\href}[2]{\texttt{#2}}
\providecommand{\urlalt}[2]{\href{#1}{#2}}
\providecommand{\doi}[1]{doi:\urlalt{https://doi.org/#1}{#1}}
\providecommand{\eprint}[1]{arXiv:\urlalt{https://arxiv.org/abs/#1}{#1}}
\providecommand{\bibinfo}[2]{#2}

\bibitemdeclare{inproceedings}{SNSAT}
\bibitem{SNSAT}
\bibinfo{author}{Fran{\c{c}}ois \surnamestart Laroussinie\surnameend}, \bibinfo{author}{Nicolas \surnamestart Markey\surnameend} \& \bibinfo{author}{Philippe \surnamestart Schnoebelen\surnameend} (\bibinfo{year}{2001}): \emph{\bibinfo{title}{Model Checking {CTL+} and {FCTL} Is Hard}}.
\newblock In \bibinfo{editor}{Furio \surnamestart Honsell\surnameend} \& \bibinfo{editor}{Marino \surnamestart Miculan\surnameend}, editors: {\slshape \bibinfo{booktitle}{Foundations of Software Science and Computation Structures}}, {\slshape \bibinfo{series}{Lecture Notes in Computer Science}} \bibinfo{volume}{2030}, \bibinfo{publisher}{Springer}, pp. \bibinfo{pages}{318--331}, \doi{10.1007/3-540-45315-6_21}.

\bibitemdeclare{inproceedings}{DKH2025}
\bibitem{DKH2025}
\bibinfo{author}{Bin \surnamestart Liu\surnameend} \& \bibinfo{author}{Yanjing \surnamestart Wang\surnameend} (\bibinfo{year}{2025}): \emph{\bibinfo{title}{Distributed Knowing How}}.
\newblock In \bibinfo{editor}{Adam \surnamestart Bjorndahl\surnameend}, editor: {\slshape \bibinfo{booktitle}{{\rm Proceedings Twentieth Conference on} Theoretical Aspects of Rationality and Knowledge, {\rm D\"usseldorf, Germany, July 14-16, 2025}}}, {\slshape \bibinfo{series}{Electronic Proceedings in Theoretical Computer Science}} \bibinfo{volume}{437}, \bibinfo{publisher}{Open Publishing Association}, pp. \bibinfo{pages}{80--97}, \doi{10.4204/EPTCS.437.11}.

\bibitemdeclare{inproceedings}{TogetherWeKnowHow2017}
\bibitem{TogetherWeKnowHow2017}
\bibinfo{author}{Pavel \surnamestart Naumov\surnameend} \& \bibinfo{author}{Jia \surnamestart Tao\surnameend} (\bibinfo{year}{2017}): \emph{\bibinfo{title}{Together We Know How to Achieve: An Epistemic Logic of Know-How (Extended Abstract)}}.
\newblock In \bibinfo{editor}{J{\'e}r{\^o}me \surnamestart Lang\surnameend}, editor: {\slshape \bibinfo{booktitle}{Proceedings of TARK 2017}}, {\slshape \bibinfo{series}{Electronic Proceedings in Theoretical Computer Science}} \bibinfo{volume}{251}, pp. \bibinfo{pages}{441--453}, \doi{10.4204/EPTCS.251.32}.

\bibitemdeclare{article}{TogetherWeKnowHow2018}
\bibitem{TogetherWeKnowHow2018}
\bibinfo{author}{Pavel \surnamestart Naumov\surnameend} \& \bibinfo{author}{Jia \surnamestart Tao\surnameend} (\bibinfo{year}{2018}): \emph{\bibinfo{title}{Together We Know How to Achieve: An Epistemic Logic of {Know-How}}}.
\newblock {\slshape \bibinfo{journal}{Artificial Intelligence}} \bibinfo{volume}{262}, pp. \bibinfo{pages}{279--300}, \doi{10.1016/j.artint.2018.06.007}.

\bibitemdeclare{inproceedings}{Wang2015}
\bibitem{Wang2015}
\bibinfo{author}{Yanjing \surnamestart Wang\surnameend} (\bibinfo{year}{2015}): \emph{\bibinfo{title}{A Logic of Knowing How}}.
\newblock In: {\slshape \bibinfo{booktitle}{Logic, Rationality, and Interaction}}, {\slshape \bibinfo{series}{Lecture Notes in Computer Science}} \bibinfo{volume}{9394}, \bibinfo{publisher}{Springer}, pp. \bibinfo{pages}{392--405}, \doi{10.1007/978-3-662-48561-3_32}.

\bibitemdeclare{article}{Wang2018}
\bibitem{Wang2018}
\bibinfo{author}{Yanjing \surnamestart Wang\surnameend} (\bibinfo{year}{2018}): \emph{\bibinfo{title}{A Logic of Goal-Directed Knowing How}}.
\newblock {\slshape \bibinfo{journal}{Synthese}} \bibinfo{volume}{195}(\bibinfo{number}{10}), pp. \bibinfo{pages}{4419--4439}, \doi{10.1007/s11229-016-1272-0}.

\end{thebibliography}

\end{document}